\documentclass[conference,letterpaper,10pt]{IEEEtran}
\onecolumn
\usepackage{etoolbox}
\makeatletter
\patchcmd{\thebibliography}{\footnotesize}{\normalsize}{}{}
\makeatother

\usepackage[utf8]{inputenc} 
\usepackage[T1]{fontenc}
\usepackage{url}
\usepackage{ifthen}
\usepackage{cite}
\usepackage[cmex10]{amsmath} 

\usepackage{amsmath,amsthm,amssymb}
\usepackage{enumerate}
\usepackage{color}
\usepackage{xcolor}
\usepackage{url}
\usepackage{balance}
\usepackage{graphicx} 
\usepackage{mathrsfs}
\usepackage{epsfig}
\usepackage{verbatim}
\usepackage{setspace}
\usepackage{cite}
\usepackage{bbm}
\usepackage{hyperref}

  \usepackage{pgfplots}
  \pgfplotsset{compat=newest}
  \usetikzlibrary{plotmarks}
  \usetikzlibrary{arrows.meta}
  \usepgfplotslibrary{patchplots}
  \usepackage{grffile}

\usepackage{multicol}
\usepackage{lipsum}
\usepackage{etoolbox}
\usepackage{algpseudocode}
\usepackage{algorithmicx}
\usepackage{algorithm}

\usepackage{pgfplots}
  \pgfplotsset{compat=newest}
  \usetikzlibrary{plotmarks}
  \usetikzlibrary{arrows.meta}
  \usepgfplotslibrary{patchplots}
  \usepackage{grffile}
  \usepackage{amsmath}

\newtheorem{theorem}{Theorem}
\newtheorem{lemma}{Lemma}

\newtheorem{corollary}{Corollary}
\newtheorem{definition}{Definition}

\newtheorem{remark}{Remark}

\newcommand{\bv}{\mathbf{v}}
\newcommand{\bx}{\mathbf{x}}

\newcommand{\by}{\mathbf{y}}

\newcommand{\bq}{\mathbf{q}}
\newcommand{\bp}{\mathbf{p}}
\newcommand{\bw}{\mathbf{w}}

\newcommand{\bu}{\mathbf{u}}

\newcommand{\cU}{{\cal U}}

\newcommand{\cX}{{\cal X}}

\newcommand{\cL}{{\cal L}}

\newcommand{\cP}{{\cal P}}

\newcommand{\cQ}{{\cal Q}}

\usepackage{pgfplots}
\pgfplotsset{compat=1.18}
\newlength\figureheight
\newlength\figurewidth
\begin{document}
\title{Geometry of Rényi Entropy on the Majorization Lattice} 


\author{%
  \IEEEauthorblockN{Anuj Kumar Yadav}
  \IEEEauthorblockA{School of Computer \& Communication Sciences \\
                    École Polytechnique Fédérale de Lausanne (EPFL)\\
                    Email: anuj.yadav@epfl.ch}
  \and
  \IEEEauthorblockN{Yanina Y. Shkel}
    \IEEEauthorblockA{School of Computer \& Communication Sciences \\
                    École Polytechnique Fédérale de Lausanne (EPFL)\\
                    Email: yanina.shkel@epfl.ch}
}

\maketitle


\begin{abstract}
 Majorization is a stochastic ordering relation that compares the relative diversity of probability distributions with numerous applications in econometrics, spectral theory and ecology. It is well-known that the majorization partial order forms a complete lattice on the set of ordered probability distributions. In this work, we study the properties of Rényi entropy on the majorization lattice. We establish a fundamental relation between the comonotone coupling and the independent coupling associated with a collection of marginal distributions. Consequently, we show that, for every order $\alpha \in [0,\infty]$, the Rényi entropy is subadditive on the majorization lattice. We further characterize the supermodular regime, showing that Rényi entropy is supermodular on the majorization lattice for $\alpha \in \{0\} \,\cup \, [1,\infty]$. For the Tsallis entropy, we show that it also satisfies subadditivity on the majorization lattice, for every order $\alpha \in [0,\infty)$. Finally we show that, unlike the Rényi entropy, the Tsallis entropy is supermodular on the majorization lattice for every $\alpha \in [0,\infty)$.
 \end{abstract}

\section{Introduction}

Majorization~\cite{marshall} is a widely studied mathematical concept for comparing the dispersion or inequality of the components of one vector relative to another. Its origin goes back to the work of Hardy, Littlewood, and P\'olya on inequalities~\cite{hardy, hardy1952inequalities}, which was further developed through Schur's study of symmetric convex functions and determinant inequalities~\cite{Schur1923}. Beyond pure inequality theory, majorization is closely connected to the Lorenz curve~\cite{lorenz_income} and the Pigou--Dalton principle of transfers~\cite{pigou,dalton}, which provide foundational tools for comparing income and wealth distributions in economics. Thus, majorization offers a natural tool for comparing concentration,
dispersion, and inequality across probability distributions and more general 
vectors.
\smallbreak
Majorization-based ideas have appeared in several domains recently. In information theory,~\cite{van} studied the relation between Rényi divergences and majorization through Lorenz diagrams and Markov ordering. \cite{catalytic} studied matrix majorization in large samples and its catalytic formulation, while Elkouss et al. derived finite sufficient conditions for catalytic majorization and thermal transformations~\cite{catalytic_2}. In wireless communications and signal processing, majorization has been used to compare eigenvalues and power allocation structures in MIMO transceiver design~\cite{mimo}. In quantum information theory, Nielsen~\cite{nielsenn} characterized
pure-state entanglement transformations using majorization, while~\cite{gour} developed a quantum-majorization framework for state transformations in quantum
thermodynamics. Majorization methods also remain central in economics and statistics with recent work on grouped-data estimation and multivariate Lorenz comparisons~\cite{jorda,fan}. In~\cite{nst}, the authors developed Newton-Simpson-type inequalities using majorization theory. It has also been applied to study the minimum entropy couplings~\cite{mec, cheuk, compton}, functional representations~\cite{mir,yadav-shkel,yadav25}, and causal inference~\cite{eci1,eci2}. In~\cite{yadav-shkel}, the authors introduced a stronger notion of an information-spectrum-based majorization (cf.~\cite{cheuk_2}) to derive information-theoretic lower bounds on minimum Rényi entropy couplings (see also~\cite{compton_aistats}).
\smallbreak
Bapat~\cite{bapat} showed that the majorization order forms a complete lattice on the set of ordered probability distributions. Entropy inequalities on lattices constitute another important line of work. More generally, submodularity and supermodularity are structural properties that provide monotonicity, convexity, and optimization principles on ordered spaces~\cite{lova,fuji2}. On the Boolean lattice formed by subsets of a finite set of random variables, Shannon entropy was shown to be submodular~\cite{fuji}. This property underlies the polymatroidal structure of Shannon entropy and many basic information inequalities~\cite{fuji,yeung}. It also plays an important role in applications such as secret sharing and the entropy methods in combinatorics~\cite{secret,jukna}. Cicalese and Vaccaro~\cite{cicalese2002supermodularity} initiated the analogous study on the majorization lattice, proving that Shannon entropy is supermodular and subadditive on this lattice. They also suggested that these inequalities may
support an information-theoretic calculus within the lattice, analogous to joint entropy, conditional entropy, and mutual information in the classical setting.
The majorization lattice has since been studied in several related directions~\cite{lattice_1,lattice_2,lattice_3,lattice_4}.


\subsection{Our Contributions}
We study structural properties of R\'enyi entropy on the majorization lattice. 
Our results recover the subadditivity and supermodularity phenomena of Shannon entropy studied in~\cite{cicalese2002supermodularity} and extend them to the Rényi entropy. Our approach is based on the analysis of couplings and Lorenz-curve geometry, yielding shorter and
more conceptual proofs than the case-by-case analysis for the Shannon entropy in~\cite{cicalese2002supermodularity}. 

\smallbreak
\noindent $\bullet$ \textbf{Relation between couplings. }
We show that for any two marginal probability mass functions (PMFs), the comonotone coupling, also known as the north-west coupling, majorizes the independent coupling. Our proof also extends to the setting of more than two marginal PMFs.
\smallbreak
\noindent $\bullet$ \textbf{Subadditivity of Rényi entropy. }
Using this coupling result, we prove that for every order $\alpha\in[0,\infty]$, the Rényi entropy is subadditive on the majorization lattice. Furthermore, we show that this property extends beyond two PMFs in the majorization lattice.
\smallbreak
\noindent $\bullet$ \textbf{Supermodularity of Rényi entropy. }
We prove that Rényi entropy is supermodular on the majorization lattice for every $\alpha\in\{0\}\,\cup\,[1,\infty]$. Furthermore, for $\alpha\rightarrow0$ and $\alpha\rightarrow\infty$, the inequality holds with equality for all pairs of PMFs; equivalently, $H_0$ and $H_\infty$ are modular on the lattice. Finally, we show that for $\alpha\in(0,1)$, Rényi entropy is neither supermodular nor submodular.
\smallbreak
\noindent $\bullet$ \textbf{Subadditivity of Tsallis entropy. } We show that, like the Rényi entropy, for every order $\alpha \in [0,\infty)$ the Tsallis entropy is subadditive on the majorization lattice. However, our proof technique substantially differs from the Rényi entropy case because Tsallis entropy does not satisfy additivity for the independent coupling, unlike the Rényi entropy. We extend the existing result on the subadditivity of the Tsallis entropy for order $\alpha \geq 1$ in~\cite{bhatia}.
\smallbreak
\noindent $\bullet$ \textbf{Supermodularity of Tsallis entropy. }
Finally, we show that, unlike the Rényi entropy, the Tsallis entropy is supermodular on the majorization lattice for every $\alpha\in[0,\infty)$. Furthermore for $\alpha=0$, the inequality holds with equality for all pairs of PMFs; equivalently, $T_0$ is modular on the lattice. This is in contrast to the submodularity of Tsallis entropy in~\cite{bhatia}, which we discuss in detail in Remark~\ref{rem:bsk}.

\section{Notations and Preliminaries}\label{sec:preliminaries}

The set of all probability mass functions (PMFs) of size $n$ is denoted by $\mathcal{P}'_n$. Further, we use $\mathcal{P}_n$ to denote the set of all ordered PMFs of size $n$, with masses in the non-increasing order. Clearly, $\mathcal{P}^{}_n \subset \mathcal{P}'_n$.
We denote PMFs by a bold-face letter, say $\mathbf{p}:=(p_1,p_2,\cdots,p_n) \in \mathcal{P}'_n$. Given a random variable $X$, its support is denoted by $\cX$, while a realization is denoted by lower case letter, for example, $x\in \cX$. We use $[n]$ to denote the set $\{1,2,\ldots,n\}$. Further, $[a,b]$ denotes the closed interval from $a$ to $b$ where $a \leq b$.

\begin{definition}[Rényi entropy]
The R\'{e}nyi entropy~\cite{renyi} of order $\alpha \in [0,\infty]$ of a random variable $X \sim \mathbf{p}$ is defined as
\begin{align}
H_{\alpha}(X) = H_{\alpha}(\mathbf{p})\triangleq \frac{1}{1-\alpha} \log\left( \sum_{i=1}^{n}\left(p_i\right)^ {\alpha}\right).
\end{align}
As $\alpha \rightarrow 1$, the Rényi entropy reduces to the well-known Shannon entropy. 
\end{definition}
\begin{definition}[Tsallis entropy]
The Tsallis entropy~\cite{tsallis} of order $\alpha \in [0,\infty)$ of a random variable $X \sim \mathbf{p}$ is defined as
\begin{align}
T_{\alpha}(X) = T_{\alpha}(\mathbf{p})\triangleq \dfrac{1-\sum^{n}_{i=1}(p_i)^{\alpha}}{\alpha-1}.
\end{align}
As $\alpha \rightarrow 1$, the Tsallis entropy also reduces to the Shannon entropy. An interesting property of Tsallis entropy which makes it different from Rényi entropy is pseudo-additivity i.e., for two independent random variables $X$ and $Y$, we have
\begin{align}
    T_{\alpha}(XY) =T_{\alpha}(X)+T_{\alpha}(Y)+(1-\alpha)T_{\alpha}(X)T_{\alpha}(Y),
\end{align}
for every $\alpha \in [0,\infty) \setminus\{1\}$.
\end{definition}

Tsallis entropy has been widely used as a generalized uncertainty measure for systems that deviate from the classical Boltzmann–Gibbs framework~\cite{tsallis2,tsallis3}, particularly in non-extensive statistical mechanics, chaotic dynamics, and systems with long-range interactions or memory effects. It has also found applications in image processing~\cite{image}, especially entropy-based image thresholding and segmentation, and in cosmological models such as Tsallis holographic dark energy~\cite{holography}. More recently, Tsallis-based divergences and entropy losses also appear in machine learning~\cite{tsd}, including supervised classification under class imbalance and adaptation methods for neural networks (cf.~\cite{tsalliscf1,tsalliscf2}).

\smallbreak
Throughout the paper, we assume that the PMFs are arranged in the non-increasing order i.e., $\mathbf{p} \in \mathcal{P}_n$. Let $\bp,\bq\in\mathcal P_n$, and define
\begin{align*}
    P_i\triangleq \sum_{k=1}^i p_k,
    \quad\text{and}\quad
    Q_j \triangleq \sum_{k=1}^j q_k,
\end{align*}
with the convention $P_0=Q_0=0$. The independent coupling
$\pi_I(\bp,\bq)$ is the joint PMF on $[n]\times[n]$ given by
\begin{align*}
    \pi_I(\bp,\bq)(i,j)\triangleq p_iq_j.
\end{align*}
The comonotone coupling $\pi_c(\bp,\bq)$ (also known as the north-west coupling) is the joint PMF
defined by
\begin{align*}
    \pi_c(\bp,\bq)(i,j)
    \triangleq
    \left[
        \min\{P_i,Q_j\}
        -
        \max\{P_{i-1},Q_{j-1}\}
    \right]_+,
\end{align*}
where $[x]_+=\max\{x,0\}$. Equivalently, $\pi_c(\bp,\bq)(i,j)$ is the length of the
overlap between the intervals $(P_{i-1},P_i]$ and $(Q_{j-1},Q_j]$.
\begin{definition}[{Aggregation $(\sqsubseteq)$}]
Let $\mathbf{q}=(q_1,\ldots,q_m)$ and $\mathbf{p}=(p_1,\ldots,p_n)$ be probability mass
functions with $m\leq n$. We say that $\mathbf{q}$ is an {aggregation} of $\mathbf{p}$, denoted by $\bp \sqsubseteq \bq$,
if there exists a partition $\{A_1,\ldots,A_m\}$ of $[n]$ such that
\begin{align}
q_i=\sum_{j\in A_i} p _j,
\end{align}
for every $i \in [m]$.
\end{definition}
\begin{definition}[Partial Order $(\preceq)$]\label{def:1}
Let $\cP$ be a set. A homogeneous binary relation on $\cP$, denoted by $\preceq\,$ $\subseteq \,\cP \times \cP$, is called a partial order if, $\forall$ $x,y,z \in \cP$, the following holds
\begin{enumerate}[(a)]
    \item \textbf{Reflexivity:} $(x,x) \in$ $\preceq$, $\forall$ $x \in \cP$.
    \item \textbf{Antisymmetry:} If $(x,y)\in$ $\preceq$ and $(y,x) \in$ $\preceq$, then $x=y$.
    \item \textbf{Transitivity:} If $(x,y)\in$ $\preceq$ and $(y,z) \in$ $\preceq$, then $(x,z) \in$ $\preceq$.
\end{enumerate}
The pair $(\cP,\preceq)$ is called the partially-ordered set (poset). In addition, if for every $x,y\in\mathcal P$, either $x\preceq y$ or $y\preceq x$, then $\preceq$ is called a {total order}.\footnote{Throughout the paper, we write $(a,b)\,\in\,\,\preceq$ equivalently as $a\,\preceq \,b$.}
\end{definition}
\begin{remark}
A binary relation that is reflexive and transitive, but not necessarily antisymmetric is called a {pre-order} or {quasi-order}.
\end{remark}
Now we define the \textit{greatest lower bound (glb)} and the \textit{least upper bound (lub)} with respect to a partial order $\preceq$.
\begin{definition}[\textit{glb} $(\wedge)$ and \textit{{lub}} $(\vee)$]
Let $(\cP,\preceq)$ be a poset, and let $\cQ \subseteq \cP$. An element  $l \in \cP$ is called a lower bound of $\cQ$ if $\forall$ $x \in \cQ$, we have $l \preceq x$. Let $$\cL(\cQ):=\{l \in \cP : l \preceq x, \text{ }\forall x \in \cQ \}$$ 
denote the set of lower bounds of $\mathcal{Q}$. If $\cL(\cQ)$ has a greatest element, then it is called the glb of $\cQ$, denoted by $\bigwedge \cQ$ i.e., $\bigwedge \cQ \in \cL(\cQ)$ and $l \preceq \bigwedge \cQ$, for every $l \in \cL(\cQ)$. 

Similarly, an element  $u \in \cP$ is called an upper bound of $\cQ$ if $\forall$ $x \in \cQ$, we have $x \preceq u$. Let $$\cU(\cQ):=\{u \in \cP : x \preceq u, \text{ }\forall x \in \cQ \}$$ 
denote the set of upper bounds of $\mathcal{Q}$. If $\cU(\cQ)$ has a least element, then it is called the lub of $\cQ$, denoted by $\bigvee \cQ$ i.e., $\bigvee \cQ \in \cU(\cQ)$ and $\bigvee \cQ \preceq u$, for every $u \in \cU(\cQ)$. 
\end{definition}
\begin{definition}[Lattice]\label{rem:pol:1}
Let $(\cP,\preceq)$ be a poset. We say that $\mathcal P$ is a
\emph{lattice} if every pair of elements $x,y\in\mathcal P$ has both a glb and a lub in $\mathcal{P}$, denoted by $x \wedge y$ and $x \vee y$, respectively. In this case, we denote the lattice by the quadruple $(\mathcal P,\preceq,\wedge,\vee)$. If every subset $\mathcal Q\subseteq\mathcal P$ has both a glb and a lub, then $\mathcal P$ is called a {complete lattice}.
Thus, every lattice is a poset, but not every poset is a lattice.
\end{definition}

\section{The Majorization Lattice}

In this section, we define the majorization partial order and state some related properties.

\begin{definition}[Majorization]\label{def:3}
Let $\mathbf{p}$ and $\mathbf{q}$ $\in$ $\cP_{n}$. Then, we say $\bq$ majorizes $\bp$, denoted by $\mathbf{p} \preceq \mathbf{q}$, if
\begin{align}
    \sum_{i=1}^{k}p_{i} \leq    \sum_{i=1}^{k}q_{i}
\end{align}
holds for every $k \in [n]$. 

If $\bp$ and $\bq$ have finite supports with different cardinalities, we pad the shorter PMF with zeros and apply the same definition. The notion of majorization can also be extended to PMFs with countably infinite supports.\footnote{Note that the majorization is not necessarily restricted to sorted PMFs. If the PMFs $\bp$ and $\bq$ are not ordered in the non-increasing order, we may compare the sums of the $k$ largest probability masses instead of the partial sum.}
\end{definition}
\begin{remark}
    Majorization is  a partial order on the set of ordered PMFs i.e., $\cP_{n}$. While it is a pre-order on the set $\cP'_{n}$. 
\end{remark}

\begin{remark}
 If a PMF $\bq$ is an aggregation of a PMF $\bp$ i.e., $\bp \sqsubseteq \bq$ then $\bq$ majorizes $\bp$ i.e., $ \bp \preceq \bq$. Thus, aggregation implies majorization. 
\end{remark}
\begin{definition}[Schur-convex function]\label{def:schurconvex}
A function $f:\cP_{n}\rightarrow \mathbb{R}$ is said to be Schur-convex if for every $\bp,\bq$ $\in$ $\cP_{n}$ satisfying $\bp \preceq \bq$, we have
\begin{align}
 f(\bp)\leq f(\bq)
\end{align}
Equivalently, $f$ is called {Schur-concave} if $-f$ is Schur-convex. R\'{e}nyi entropy $(H_{\alpha}(\bp))$ as well as Tsallis entropy $T_{\alpha}(\bp)$ of any order $\alpha\geq0$ are Schur-concave functions.
\end{definition}
  \begin{remark}
Bapat~\cite{bapat} showed that the poset $(\mathcal P_n,\preceq)$ induced by majorization is a complete lattice. Cicalese and Vaccaro~\cite{cicalese2002supermodularity} gave explicit constructions for the greatest lower bound (glb) and the least upper bound (lub) of any two elements. 

\noindent Given $\bp$ and $\bq$ in $\mathcal{P}_n$, the glb $(\bp \wedge \bq)$ is given by
\begin{align}
    \sum^{k}_{i=1} (\bp \wedge \bq)(i) = \min\left\{\sum_{i=1}^{k}p_i,\sum_{i=1}^{k}q_i \right\}\label{eq:glb}
\end{align}
Define a vector $\beta_{\bp,\bq} \in \mathbb{R}^{n}$ such that
\begin{align}
   \sum^{k}_{i=1} \beta_{\bp,\bq}(i)=\max\left\{\sum_{i=1}^{k}p_i,\sum_{i=1}^{k}q_i \right\} \label{eq:lub}
\end{align}
If $\beta_{\bp,\bq}\in\mathcal P_n$, then the lub $\bp\vee \bq=\beta_{\bp,\bq}$. Otherwise, $\bp \vee \bq$ is obtained by repeatedly replacing each maximal consecutive block of coordinates of $\beta_{\bp,\bq}$ that violates the non-increasing order by its average, until the resulting vector lies in $\mathcal P_n$. Thus, $\bp\vee \bq$ is the least concave majorant (LCM) of $\beta_{\bp,\bq}$.
\end{remark}
Lorenz curves~\cite{lorenz_income,marshall} provide a geometric representation of PMFs and are particularly useful in the study of majorization: the majorization order is equivalently characterized by pointwise comparison of the corresponding Lorenz curves. We use this geometric perspective in the sequel to establish our results.
\begin{definition}[Lorenz curve]\label{def:lorenz}
Let $\bp=(p_1,\ldots,p_n)\in\mathcal P_n$. The {Lorenz curve} of $\bp$ is the piecewise-linear concave function $L_\bp:[0,n]\to[0,1]$ obtained by linear interpolation of the points
\[
\left(k,\sum_{i=1}^k p_i\right),
\qquad k \in [n],
\]
where $\sum_{i=1}^0 p_i:=0$. Equivalently, $L_\bp(0)=0$ and, for
$x\in[k-1,k]$ with $k\in[n]$,
\[
L_\bp(x)
=
\sum_{i=1}^{k-1}p_i + (x-k+1)p_k .
\]
\end{definition}

\begin{lemma}\label{lem:ace1}
Let $\bp, \bq\in\mathcal P_n$, and let \(\pi_c(\bp,\bq)\) denote their comonotone coupling . Then, the glb and the lub of \(\bp\) and \(\bq\) in the majorization lattice are both canonically determined by \(\pi_c(\bp,\bq)\). More precisely, \(\bp\wedge \bq\) is an aggregation of the comonotone coupling \(\pi_c(\bp,\bq)\), while \(\bp\vee \bq\) is obtained by applying the least concave majorant (LCM) operation to an aggregation induced by \(\pi_c(\bp,\bq)\).
\end{lemma}
\begin{proof}
    The proof is deferred to Appendix~\ref{app:ace1}.
\end{proof}
\begin{lemma}\label{lem:ace2}
   Let $\bp, \bq\in\mathcal P_n$. Let $\pi_c(\bp,\bq)$ and $\pi_I(\bp,\bq)$ denote the comonotone coupling and the independent coupling of $\bp$ and $\bq$, respectively. Then, the comonotone coupling vector $\pi^{\downarrow}_c$ majorizes the independent coupling vector $\pi^{\downarrow}_I$ i.e., 
      \begin{align}
    \pi^{\downarrow}_I(\bp,\bq) \, \preceq \, \pi^{\downarrow}_c(\bp,\bq).
      \end{align}
\end{lemma}

\begin{proof}
    The proof is deferred to Appendix~\ref{app:ace2}.
\end{proof}

\section{Subadditivity Property}

\begin{definition}
    A real-valued function $f:\mathcal{P} \rightarrow \mathbb{R}$ defined on a lattice $(\mathcal{P},\preceq,\wedge,\vee)$ is subadditive if $\forall$ $x$, $y$ $\in$ $\mathcal{P}$, we have
    \begin{align}
        f(x \wedge y) \leq f(x)+f(y).
    \end{align}
\end{definition}
In the next theorem, we state a result on the subadditivity of Rényi entropy on the majorization lattice.
\begin{theorem}\label{thm:subadd}
    For every $\alpha\in[0,\infty]$, the Rényi entropy $H_\alpha$ is subadditive
on the majorization lattice $(\mathcal P_n,\preceq,\wedge,\vee)$. That is, for
all $\bp, \bq\in\mathcal P_n$,
    \begin{align}\label{eq:subadd}
        H_{\alpha}(\bp \wedge \bq)\leq H_{\alpha}(\bp)+H_{\alpha}(\bq).
        \end{align}
The equality in~\eqref{eq:subadd} holds if and only if at least one of $p$ and $q$ is the
deterministic PMF $(1,0,\ldots,0)$.
\end{theorem}
\begin{proof}
From Lemma~\ref{lem:ace1}, we know that $  \pi_c(\bp,\bq) \sqsubseteq \bp \wedge \bq$. Since aggregation implies majorization, we have that
\begin{align}
\pi^{\downarrow}_c(\bp,\bq) \, \preceq \, \bp \wedge \bq
\end{align}
Furthermore, using Lemma~\ref{lem:ace2} we obtain
\begin{align}
 \pi^{\downarrow}_I(\bp,\bq)\, \preceq \, \pi^{\downarrow}_c(\bp,\bq)  \, \preceq \bp \wedge \bq.
\end{align}
Since, Rényi entropy is Schur-concave for every $\alpha \in [0,\infty]$,
\begin{align}
    H_{\alpha}(\bp \wedge \bq) &\leq H_{\alpha}(\pi^{\downarrow}_I(\bp, \bq)) \label{eq:alk1}\\
    &= H_{\alpha}(\bp)\, +\,H_{\alpha}(\bq) \label{eq:alk2}
\end{align}
The above proof can also be obtained without invoking Lemma~\ref{lem:ace2}. The independent coupling $\pi_I(\bp,\bq)$ has marginals $\bp$
and $\bq$; hence, by aggregation,
\begin{align*}
    \pi_I^\downarrow(\bp,\bq)\, \preceq \,\bp
    \qquad \text{and} \qquad
    \pi_I^\downarrow(\bp,\bq)\, \preceq \,\bq,
\end{align*}
Thus $\pi_I^\downarrow(\bp,\bq)$
is a lower bound of $\bp$ and $\bq$. By the definition of the
greatest lower bound,
\begin{align}
    \pi_I^\downarrow(\bp,\bq) \,\preceq \,\bp \wedge \bq.
\end{align}
Therefore,~\eqref{eq:alk1} and~\eqref{eq:alk2} follow directly. Nevertheless, Lemma~\ref{lem:ace2} provides a sharper structural comparison between the independent and comonotone couplings, and we use it here to emphasize this relation.\\

\noindent \textbf{(Conditions for equality).} Now, it remains to precisely characterize the equality cases. First, if one of
$\bp$ and $\bq$ is deterministic, say
$\bp=(1,0,\ldots,0)$, then we have $\bp\wedge \bq=\bq$. Consequently,
\begin{align}
    H_\alpha(\bp\wedge \bq)
    =
    H_\alpha(\bq)
    =
    H_\alpha(\bp)+H_\alpha(\bq),
\end{align}
since $H_\alpha(\bp)=0$. The similar argument applies if
$\bq$ is deterministic.

Conversely, now assume that neither $\bp$ nor $\bq$ is
deterministic. Let us suppose
\begin{align}
    a := |\operatorname{supp}(\bp)|,
    \quad \text{and}\quad 
    b := |\operatorname{supp}(\bq)|.
\end{align}
such that $n \geq a,b \geq 2$. Thus, the independent coupling has support size
\begin{align}
    |\operatorname{supp}(\pi^{\downarrow}_I(\bp,\bq))| = ab,
\end{align}
whereas, by Lemma~\ref{lem:ace1},
\begin{align}
    |\operatorname{supp}(\bp \wedge \bq)|
    \leq \ \max(a,b).
\end{align}
Since $ab>\max(a,b)$, the zero-padded PMF
$\bp\wedge \bq$ cannot be a permutation of
$\pi^{\downarrow}_I(\bp,\bq)$.

For $\alpha \in (0,\infty)$, the Rényi entropy $H_\alpha$ is strictly
Schur-concave. Therefore, since
\begin{align}
    \pi_I^\downarrow(\bp,\bq)
    \preceq
    \bp\wedge \bq
\end{align}
and the two vectors are not permutations of each other, we obtain the
strict inequality
\begin{align}
    H_\alpha(\bp\wedge \bq)
    <
    H_\alpha(\pi_I^\downarrow(\bp,\bq))
    =
    H_\alpha(\bp)+H_\alpha(\bq).
\end{align}
Therefore, the equality is impossible for $0<\alpha<\infty$ unless either
$\bp$ or $\bq$ is deterministic. Now, let's analyze the edge cases of $\alpha \rightarrow 0$ and $\alpha \rightarrow \infty$.

For $\alpha \rightarrow 0$, the same support-size argument gives
\begin{align}
    H_0(\bp\wedge \bq)
    &=
    \log |\operatorname{supp}(\bp\wedge \bq)|\\
    &\leq
    \log \max(a,b)\\
    &<\log(ab)\\
    &= H_0(\bp)+H_0(\bq).
\end{align}

Finally, for $\alpha \rightarrow \infty$, let
\begin{align}
    p_1=\max_i p_i,
    \qquad
    q_1=\max_i q_i.
\end{align}
Since neither distribution is deterministic, $p_1<1$ and $q_1<1$. Moreover,
by the definition of $\bp\wedge \bq$,
\begin{align}
    (\bp\wedge \bq)(1)=\min(p_1,q_1).
\end{align}
Since $\bp \wedge \bq \in \mathcal{P}_n$ is arranged in non-increasing order,
\begin{align}
    \max_i(\bp\wedge \bq)(i)
    =
    \min(p_1,q_1)
    >
    p_1q_1.
\end{align}
Consequently,
\begin{align}
    H_\infty(\bp\wedge \bq)
    <
    -\log(p_1q_1)
    =
    H_\infty(\bp)+H_\infty(\bq).
\end{align}
Therefore, equality holds if and only if at least one of
$\bp$ and $\bq$ is deterministic. This completes our proof.
\end{proof}

Furthermore, the Schur-concavity of Rényi entropy yields the following simple upper bound,
\begin{corollary}\label{cor:schur}
    For all $\bp, \bq\in\mathcal P_n$ and $\alpha \in [0,\infty]$, we have
    \begin{align}
        H_\alpha(\bp)+H_\alpha(\bq)\le 2H_\alpha(\bp\wedge \bq).
    \end{align}
For $\alpha\in(0,\infty)$, the equality holds if and only if
$p=q$.
\end{corollary}
As a consequence of Theorem~\ref{thm:subadd} and Corollary~\ref{cor:schur}, we have the following corollary which extends the above property to any finite number of PMFs,
\begin{corollary}\label{cor:subadd}
    Let $\{\bp_i\}^{m}_{i=1}$ be $m$ PMFs in $\mathcal{P}_n$. Then, for Rényi entropy of any order $\alpha \in [0, \infty]$, we have that
    \begin{align}
        H_{\alpha}\left(\bigwedge^{m}_{i=1}\bp_i\right) \leq \sum^{m}_{i=1}H_{\alpha}(\bp_i)\leq m H_{\alpha}\left(\bigwedge^{m}_{i=1}\bp_i\right).
    \end{align}
\end{corollary}
\begin{proof}
The first inequality follows by applying Theorem~\ref{thm:subadd} (on the pair of PMFs) iteratively to the glb $\bigwedge_{i=1}^m \mathbf p_i$, for $m-1$ times. 

For the second inequality, since $\bigwedge_{i=1}^m \bp_i$ is a lower bound, w.r.t majorization, for every $\bp_i$, we
have $\bigwedge_{i=1}^m \mathbf p_i \,\preceq \,\bp_i$, for all $i$. By
Schur-concavity of Rényi entropy,
$H_\alpha(\mathbf p_i)\leq H_\alpha(\bigwedge_{i=1}^m \bp_i)$ for every
$i$, and summing over $i\in[m]$ gives us the result.
\end{proof}

It was shown in~\cite{bhatia} that Tsallis entropy is subadditive for order $\alpha \geq 1$. In the next theorem, we show that the subadditivity property extends to all $\alpha \in [0,\infty)$ using our proof technique based on couplings.

\begin{theorem}\label{thm:subadd_tsallis}
    For every $\alpha\in[0,\infty)$, the Tsallis entropy $T_\alpha$ is subadditive
on the majorization lattice $(\mathcal P_n,\preceq,\wedge,\vee)$. That is, for
all $\bp, \bq\in\mathcal P_n$,
    \begin{align}\label{eq:subadd_t}
        T_{\alpha}(\bp \wedge \bq)\leq T_{\alpha}(\bp)+T_{\alpha}(\bq).
        \end{align}
The equality in~\eqref{eq:subadd_t} holds if and only if at least one of $p$ and $q$ is the
deterministic PMF $(1,0,\ldots,0)$.
\end{theorem}
\begin{proof}
Unlike the Rényi entropy, the Tsallis entropy only satisfies pseudoadditivity i.e., given PMFs $\bp$ and $\bq$, we have that
\begin{align}
    T_{\alpha}(\pi^{\downarrow}_I(\bp,\bq)) = T_{\alpha}(\bp)+T_{\alpha}(\bq)+(1-\alpha)T_{\alpha}(\bp)T_{\alpha}(\bq)
\end{align}
for every $\alpha \in [0,\infty)$.  Thus, the relation between the independent coupling and comonotone coupling can be leveraged to proof subadditivity on the majorization lattice only for $\alpha \geq 1$. However, we bypass this relation to prove a stronger inequality showing that for all $\alpha \in [0,\infty)$, we have $T_{\alpha}(\pi^{\downarrow}_c(\bp,\bq)) \leq T_{\alpha}(\bp)+T_{\alpha}(\bq)$. The detailed proof is deferred to the Appendix~\ref{app:subadd_tsallis}.
\end{proof}

Similar to the case of Rényi entropy in Corollary~\ref{cor:subadd}, the extension of subadditivity to the finite collection of PMFs also holds for the Tsallis entropy of any order $\alpha \in [0,\infty)$.

\section{Supermodularity Property}

\begin{definition}
    A real-valued function $f:\mathcal{P} \rightarrow \mathbb{R}$ defined on a lattice $(\mathcal{P},\preceq,\wedge,\vee)$ is said to be supermodular if $\forall$ $x$, $y$ $\in$ $\mathcal{P}$, we have
    \begin{align}
        f(x)+f(y) \leq f(x \wedge y) + f(x \vee y).
    \end{align}
Equivalently, $f$ is submodular iff $-f$ is supermodular.
\end{definition}
In the next theorem, we state a result on the supermodularity of Rényi entropy on the majorization lattice,
\begin{theorem}\label{thm:supermod}
  For every $\alpha\in\{0\}\cup[1,\infty]$, the R\'enyi entropy $H_\alpha$ is
supermodular on the majorization lattice
$(\mathcal P_n,\preceq,\wedge,\vee)$. That is, for all
$\bp, \bq\in\mathcal P_n$,
    \begin{align}
        H_{\alpha}(\bp)+H_{\alpha}(\bq) \leq  H_{\alpha}(\bp \wedge \bq) + H_{\alpha}(\bp \vee \bq).
        \end{align}
Moreover, for $\alpha\in\{0,\infty\}$, equality holds for all $\bp, \bq\in\mathcal P_n$.\end{theorem}
\begin{proof}
    We leverage the idea of Lorenz curves to represent the glb and lub, to prove our result. We define the following shorthands: $\bw:= \bp \wedge \bq$, $\bv:=\bp \vee \bq$, and $S_{\alpha}(\bp):=\sum^{n}_{i=1} p^{\alpha}_{i}$. Thus, we need to show that for $\alpha\geq 1$,
    \begin{align}
       \hspace*{-3mm}\frac{1}{1-\alpha} [\log(S_{\alpha}(\bp))+\log(S_{\alpha}(\bq))]&\leq \frac{1}{1-\alpha}[\log(S_{\alpha}(\bw))+\log(S_{\alpha}(\bv))] 
    \end{align}
i.e., $S_{\alpha}(\bp)S_{\alpha}(\bq)\geq S_{\alpha}(\bw)S_{\alpha}(\bv)$. From Definition~\ref{def:lorenz} of Lorenz curve, we note that for every $\bx \in \mathcal{P}_n$,
\begin{align}
S_{\alpha}(\bx) = \int^{n}_{0}L'_{\bx}(t)^{\alpha}dt 
\end{align}
Let $L_{\bp}$ and $L_{\bq}$ denote the Lorenz curves of $\bp$ and $\bq$, respectively. Then, we have that for all $k \in [n]$, the Lorenz curves of $\bw$ and $\bv$
\begin{align}
    L_{\bw}(k) &= \min\{L_{\bp}(k),L_{\bq}(k)\},\\
    L_{\bv}(k) &= \mathrm{LCM}(F_{\bv}(k)),
\end{align}
where $F_{\bv}(k):=\max\{L_{\bp}(k),L_{\bq}(k)\}$ and $\mathrm{LCM(\cdot)}$ denotes the least concave majorant operation. Further, let $m(t)$ and $M(t)$ be the continuous envelopes defined for all $t \in [0,n]$,
\begin{align}
        m(t) &= \min\{L_{\bp}(t),L_{\bq}(t)\},\\
    M(t) &= \max\{L_{\bp}(t),L_{\bq}(t)\}.
\end{align}
Within any interval of integer points or crossings between $L_\bp$ and $L_{\bq}$, we have that
\begin{align*}
    m &=L_{\bp},\;\; M=L_{\bq}, \;\text{ if }\; L_{\bp}\leq L_{\bq}\\
    m &=L_{\bq},\;\; M=L_{\bp}, \;\text{ if }\; L_{\bq}\leq L_{\bp}.
\end{align*}
Thus, within any continuous interval $(m'(t),M'(t))$ is a permutation of $(L'_{\bp}(t),L'_{\bq}(t))$. Consequently, for all $t \in [0,n]$,
\begin{align}
    (m'(t))^{\alpha}+ (M'(t))^{\alpha} = (L'_{\bp}(t))^{\alpha}+(L'_{\bq}(t))^{\alpha}.
\end{align}
On integrating both sides with respect to $t \in [0,n]$,
\begin{align}\label{eq:rem}
    \int^{n}_{0}(m'(t))^{\alpha}dt\;+ \int^{n}_{0}(M'(t))^{\alpha}dt\; = S_{\alpha}(\bp) +S_{\alpha}(\bq).
    \end{align}
Fix an interval $I_k \triangleq (k-1,k)$. Then, $m$ is a piece-wise linear concave function within $I_k$ while $L_{\bw}$ is a straight line joining $m(k-1)$ to $m(k)$. Thus, 
\begin{align*}
    L'_{\bw}(t)&=\frac{m(k)-m(k-1)}{1}\\&=\int^{k}_{k-1}m'(t)dt
\end{align*}
Since $u \rightarrow u^{\alpha}$ is a convex map for $\alpha \geq 1$, using Jensen's inequality we have
\begin{align}\label{eq:main1}
     (L'_{\bw}(t))^{\alpha} \leq \int^{k}_{k-1}m'(t)^{\alpha}dt
\end{align}
Now, considering $M$ in the interval $I_k$. Unlike $m$, $M$ is a piece-wise linear (not necessarily concave) function while $F_{\bv}$ is a straight line joining $M(k-1)$ to $M(k)$. If $I_k$ has no non-integer crossings, then $F'_{\bv}(t)=M'(t)$. 

However, if there are crossings then $M$ is convex and we need to concavify it to get $F_{\bv}$. Thus,
\begin{align*}
        F'_{\bv}(t)&=\frac{M(k)-M(k-1)}{1}\\&=\int^{k}_{k-1}M'(t)dt
\end{align*}
Since $u \rightarrow u^{\alpha}$ is a convex map for $\alpha>1$, using Jensen's inequality we have
\begin{align}\label{eq:main2}
     (F'_{\bv}(t))^{\alpha} \leq \int^{k}_{k-1}(M'(t))^{\alpha}dt
\end{align}
Now, if $F_{\bv}$ is a piecewise-linear concave function, then it is the lub $L_{\bv}$. Otherwise, it needs to be concavified and consequently the the lub $L_{\bv}$ is the least concave majorant (LCM) of $F_{\bv}$.

Therefore applying Jensen again over the pooled intervals used by the LCM operation, we get
\begin{align}\label{eq:main3}
S_\alpha(\bv)
    &=
    \int_0^n (L_{\bv}'(t))^\alpha\,dt\\
    &\le
    \int_0^n (F_{\mathbf v}'(t))^\alpha\,dt
    \le
    \int_0^n (M'(t))^\alpha\,dt .
\end{align}
On summing~\eqref{eq:main1} over all intervals $I_k$ for $k \in [n]$ and using~\eqref{eq:main3}, we obtain
\begin{align}
    S_{\alpha}(\bw)+S_{\alpha}(\bv) \leq \int^{n}_{0}(m'(t))^{\alpha}dt +\int^{n}_{0}(M'(t))^{\alpha}dt 
\end{align}
Therefore, from~\eqref{eq:rem} we have
\begin{align}\label{eq:rem2}
     S_{\alpha}(\bw)+S_{\alpha}(\bv) \leq S_{\alpha}(\bp)+S_{\alpha}(\bq).
\end{align}
Recall that $\bw \preceq \bp,\;\bq \preceq \bv$. Since $S_{\alpha}$ is schur-convex for $\alpha \geq 1$, it implies that $S_{\alpha}(\bw)\leq S_{\alpha}(\bp),S_{\alpha}(\bq) \leq S_{\alpha}(\bv)$. Let $x,\;y\;,z\;\geq0$ such that 
\begin{align*}
    S_{\alpha}(\bp)=S_{\alpha}(\bw)+x\\
        S_{\alpha}(\bq)=S_{\alpha}(\bw)+y\\
    S_{\alpha}(\bv)=S_{\alpha}(\bw)+z
\end{align*}
Then, from~\eqref{eq:rem2} we have
\begin{align}\label{eq:rem3}
x+y \geq z
\end{align}
Now, 
\begin{align}
    S_{\alpha}(\bp)S_{\alpha}(\bq)&-S_{\alpha}(\bw)S_{\alpha}(\bv)\notag \\ &\hspace*{-15mm}= (S_{\alpha}(\bw)+x)(S_{\alpha}(\bw)+y)-S_{\alpha}(\bw)(S_{\alpha}(\bw)+z)\\
    &\hspace*{-15mm}=S_{\alpha}(\bw)(x+y-z) +xy \, \,
 \geq\,\, 0.  \label{eq:asdfgh}
\end{align}
where~\eqref{eq:asdfgh} follows from~\eqref{eq:rem3} and noting that $x,\;y\;\geq 0$. The case $\alpha\rightarrow1$ follows by taking the limit $\alpha\downarrow 1$, since
Rényi entropy is continuous in $\alpha$ on finite alphabets. 
\smallbreak
\noindent \textbf{Modularity for $\alpha \rightarrow 0$ and $\alpha\rightarrow\infty$.} Let
$ a := |\operatorname{supp}(\bp)|$ and $b := |\operatorname{supp}(\bq)|$,
where $a,b\, \leq \, n$. For $\alpha\rightarrow0$, we have
\begin{align*}
    |\operatorname{supp}(\bp\wedge\bq)|=\max(a,b),
    \quad
    |\operatorname{supp}(\bp\vee\bq)|=\min(a,b).
\end{align*}
Therefore,
\begin{align}
    H_0(\bp\wedge\bq)
    +
    H_0(\bp\vee\bq)
    &=
    \log\max(a,b)
    +
    \log\min(a,b)\\
    &=
    \log a+\log b\\
    &=
    H_0(\bp)+H_0(\bq).
\end{align}
\noindent Now consider $\alpha\rightarrow\infty$. Since $\bp$ and $\bq$ are arranged in non-increasing order, thus
\begin{align}
    p_1=\max_i p_i,
    \qquad
    q_1=\max_i q_i.
\end{align}
Then
\begin{align*}
    (\bp\wedge\bq)(1)=\min(p_1,q_1),
    \quad
    (\bp\vee\bq)(1)=\max(p_1,q_1).
\end{align*}
Consequently, we have that
\begin{align}
\begin{aligned}
    H_\infty(\bp\wedge\bq)
    +
    H_\infty(\bp\vee\bq)
    &=
    -\log\min(p_1,q_1)
    -
    \log\max(p_1,q_1) \\
    &=
    -\log p_1-\log q_1 \\
    &=
    H_\infty(\bp)+H_\infty(\bq).
\end{aligned}
\end{align}
Thus, the Rényi entropy is modular for $\alpha \in \{0, \infty\}$. This completes our proof.
\end{proof}

\begin{remark}
Unlike Corollary~\ref{cor:subadd}, the supermodularity property does not extend in the same way to more than two PMFs. In particular,
for $\bp_1, \bp_2, \bp_3\in\mathcal P_n$, one `does not' generally have
\begin{align}
\sum_{i=1}^3 H_\alpha(\bp_i)
\le
H_\alpha\left(\bigwedge^{3}_{i=1} \bp_i\right)
+
H_\alpha\left(\bigvee^{3}_{i=1} \bp_i\right).\notag
\end{align}
For example, on taking $\bp_1=\bp_2=\bp_3=\bp$ for any non-deterministic PMF $\bp$, the left-hand
side equals $3H_\alpha(\bp)$, whereas the right-hand side equals $2H_\alpha(\bp)$.
Thus the inequality fails whenever $H_\alpha(\bp)>0$.
\end{remark}

\begin{theorem}\label{thm:countersupermod}
For every $\alpha\in(0,1)$, the R\'enyi entropy $H_\alpha$ is neither supermodular nor submodular on the majorization lattice $(\mathcal P_n,\preceq,\wedge,\vee)$.
\end{theorem}
\begin{proof}
    The proof is deferred to the Appendix~\ref{app:countersupermod}.
\end{proof}
In the next theorem, we state a result on the supermodularity of Tsallis entropy on the majorization lattice,
\begin{theorem}\label{thm:supermod_tsallis}
  For every $\alpha\in [0,\infty)$, the Tsallis entropy $T_\alpha$ is
supermodular on the majorization lattice
$(\mathcal P_n,\preceq,\wedge,\vee)$. That is, for all
$\bp, \bq\in\mathcal P_n$,
    \begin{align}
        T_{\alpha}(\bp)+T_{\alpha}(\bq) \leq  T_{\alpha}(\bp \wedge \bq) + T_{\alpha}(\bp \vee \bq).
        \end{align}
Moreover, for $\alpha=0$, equality holds for all $\bp, \bq\in\mathcal P_n$.\end{theorem}
\begin{proof}
    The proof is based on the idea of Lorenz curves, similar to the one in Theorem~\ref{thm:supermod}. We defer the details to the Appendix~\ref{app:supermod_tsallis}.
\end{proof}
\begin{remark}\label{rem:bsk}
It is stated in~\cite{bhatia} (Theorem~$3.2$) that the Tsallis entropy is submodular on the majorization lattice, for every order $\alpha \in [0,\infty)$. This contradicts with our Theorem~\ref{thm:supermod_tsallis}. We note that the
discrepancy is one of the signs in the proof in~\cite{bhatia}. The argument therein controls the power sum
$S_\alpha(\bp) := \sum_{i=1}^{n} p_i^{\alpha}$, which is submodular on the lattice for $\alpha > 1$ and supermodular for $\alpha < 1$ (by convexity, resp.\ concavity, of $u \mapsto u^{\alpha}$). The Tsallis entropy is the affine transform of $S_{\alpha}(\cdot)$ i.e., $T_\alpha = \frac{1 - S_\alpha}{\alpha - 1}$, whose slope $\frac{-1}{\alpha-1}$ changes sign at $\alpha = 1$. This sign change exactly compensates the change in the modularity direction of $S_\alpha$. Consequently, $T_\alpha$ is supermodular for every $\alpha \in [0, \infty)$ --- rather than submodular.
\smallbreak
As an explicit example, consider $\bp = (0.6, 0.2, 0.2)$ and
$\bq = (0.45, 0.4, 0.15)$, such that
$\bp \wedge \bq = (0.45, 0.35, 0.2)$ and
$\bp \vee \bq = (0.6, 0.25, 0.15)$. On both sides of $\alpha=1$, we have that
\begin{align*}
  T_2(\bp \wedge \bq) + T_2(\bp \vee \bq)
    &= 1.190 > 1.175
     = T_2(\mathbf{p}) + T_2(\mathbf{q}), \\
  T_{1/2}(\bp \wedge \bq) + T_{1/2}(\bp \vee \bq)
    &= 2.743 > 2.719
     = T_{1/2}(\mathbf{p}) + T_{1/2}(\mathbf{q}),
\end{align*}
each violating submodularity as stated in~\cite{bhatia} and therefore is consistent with our
Theorem~\ref{thm:supermod_tsallis}.
\end{remark}

\section{Applications to Econometrics}

Let $\bx,\by\in\mathcal P_n$. Inspired by the work in~\cite{cica-eco}, we define
\begin{align}
    d_\alpha(\bx,\by)\triangleq    H_\alpha(\bx)+H_\alpha(\by)-    2H_\alpha(\bx\vee\by),
\end{align}
where $H_\alpha(\cdot)$ denotes the Rényi entropy of order $\alpha$. Using
Theorem~\ref{thm:subadd} (subadditivity) and Theorem~\ref{thm:supermod} (supermodularity), we can show that $d_\alpha(\cdot,\cdot)$ is a metric on $\mathcal P_n$ for
$\alpha\geq 1$. In particular, for $\alpha \rightarrow1$, this recovers the
Shannon-entropy based distance
\begin{align}
    d(\bx,\by) \triangleq H(\bx)+H(\by)-2H(\bx\vee\by),
\end{align}
studied in~\cite{cica-eco}.

This metric family also admits a natural interpretation as a measure of inequality. Let
$\bu_n=(1/n,\ldots,1/n)$ denote the uniform PMF. Since
the lub $\bx\vee\bu_n=\bx$, we have
\begin{align}
    d_1(\bx,\bu_n) = \log n-H(\bx),
\end{align}
which is precisely the Theil index~\cite{theil} of the population $\bx$. More
generally,
\begin{align}
    d_\alpha(\bx,\bu_n) = \log n-H_\alpha(\bx)
\end{align}
provides a Rényi-parametrized analogue of the Theil index, measuring the deviation of $\bx$ from the egalitarian population. The parameter $\alpha \in [1,\infty)$ allows this deviation from uniformity to be assessed at different levels of sensitivity. In particular, larger values of $\alpha$ place greater
emphasis on the largest probabilities, so the resulting distance would penalize concentration in dominant components more sharply than the Shannon-case. This suggests that the choice of $\alpha$ may be useful in applications where different notions of inequality or concentration are of
interest. We believe that a systematic study of these distances, their statistical
properties, and empirical behavior would be an interesting direction for future work.

\section{Conclusion}
In this work, we studied subadditivity and supermodularity properties of
Rényi entropy and the Tsallis entropy on the majorization lattice. Although, both entropy families are closely related, our results show that they exhibit quite different properties on the majorization lattice. As part of our analysis, we established a fundamental result in coupling theory showing that for any collection of marginal PMFs, the comonotone coupling majorizes the independent coupling. Leveraging Schur-concavity, this provides a simple and structural route to Rényi and Tsallis entropy inequalities on the majorization lattice.
\smallbreak
Beyond these lattice-theoretic results, we introduced an information-theoretic distance between PMFs which serves as a Rényi-parameterized analogue of the Theil index. More broadly, since majorization and Rényi entropies arise naturally in econometrics, source coding, guessing problems, and quantum information tasks such as entanglement conversion, we believe that our results may have broader applications in economics, information theory, and quantum information. Finally, since Tsallis entropy arise in nonextensive thermodynamics and statistical physics, the contrast between the Rényi and Tsallis entropy cases observed here may help clarify how different entropy families capture dependence, aggregation, and disorder in applied settings.

\section*{Acknowledgment}
This work is supported by the Swiss National Science Foundation (SNSF) grant number $211337$.



\newpage

\bibliographystyle{alpha}
\bibliography{IEEEabrv,ref}
\clearpage
\newpage
\appendix
\subsection{Proof of Lemma~\ref{lem:ace1}}\label{app:ace1}
\begin{proof}
Given $\bp,\,\bq\, \in \mathcal{P}_n$. Let $U$ be a uniformly distributed continuous random variable on $[0,1]$ i.e., $U \sim \mathrm{Unif}([0,1])$. Now, define random variables $X$ and $Y$ on the support set $\{1,2,\dots,n\}$ as
 \begin{align*}  
    X &\triangleq \min\left\{i:\sum_{j=1}^{i}p_j \geq U\right\}  \\   Y &\triangleq \min\left\{i:\sum_{j=1}^{i}q_j \geq U\right\}
\end{align*}
Thus, we have that $X \sim \bp$, $Y \sim \bq$ and they are coupled through the random variable $U$. Thus, the joint distribution of $(X,Y)$ is the comonotone coupling of $\bp$ and $\bq$. 

Now, let us define $Z \triangleq \max(X,Y)$ and $W \triangleq \min(X,Y)$. We will show that $Z \sim \bp \wedge \bq$ and that the lub $\bp \vee \bq$ can be obtained by the application of LCM operation on the PMF of $W$. We have,
\begin{align}
    \mathbb{P}(Z \leq i) &= \mathbb{P}(X \leq i, Y \leq i)\\
    &=\mathbb{P}\left(U \leq \sum^{i}_{j=1}p_j,U \leq \sum^{i}_{j=1}q_j\right)\\
 &=\mathbb{P}\left(U \leq \min\left\{\sum^{i}_{j=1}p_j, \sum^{i}_{j=1}q_j\right\}\right)\\    
 &=\min\left\{\sum^{i}_{j=1}p_j, \sum^{i}_{j=1}q_j\right\}
\end{align}
Therefore, 
\begin{align}
    \mathbb{P}(Z = i)    
 &=\min\left\{\sum^{i}_{j=1}p_j, \sum^{i}_{j=1}q_j\right\}-\min\left\{\sum^{i-1}_{j=1}p_j, \sum^{i-1}_{j=1}q_j\right\}
\end{align}
From~\eqref{eq:glb}, we have that $Z \sim \bp \wedge \bq $. Consequently, for all $k \in [n]$,
\begin{align}
\bp \wedge \bq \,(k) = \sum_{(i,j):\, \max(i,j)=k} \pi_c(\bp,\bq)(i,j)
\end{align}
Similarly, we have that
\begin{align}
    \mathbb{P}(W \leq i) &= \mathbb{P}(X \leq i \,\, \mathrm{or } \,\,Y \leq i)\\
    &=\mathbb{P}\left(U \leq \sum^{i}_{j=1}p_j\,\, \mathrm{or}\, \,U \leq \sum^{i}_{j=1}q_j\right)\\
 &=\mathbb{P}\left(U \leq \max\left\{\sum^{i}_{j=1}p_j, \sum^{i}_{j=1}q_j\right\}\right)\\    
 &=\max\left\{\sum^{i}_{j=1}p_j, \sum^{i}_{j=1}q_j\right\}
\end{align}
Therefore, 
\begin{align}
    \mathbb{P}(W = i)    
 &=\max\left\{\sum^{i}_{j=1}p_j, \sum^{i}_{j=1}q_j\right\}-\max\left\{\sum^{i-1}_{j=1}p_j, \sum^{i-1}_{j=1}q_j\right\}
\end{align}
From~\eqref{eq:lub}, we have that $W \sim \beta_{\bp,\bq}$. Consequently, for all $k \in [n]$,
\begin{align}
\beta_{\bp,\bq}\,(k) = \sum_{(i,j):\, \min(i,j)=k} \pi_c(\bp,\bq)(i,j)
\end{align}
The lub $\bp \vee \bq$ can be obtained by applying the least concave majorant operation on $\beta_{\bp,\bq}$. This completes our proof.
\end{proof}
\subsection{Proof of Lemma~\ref{lem:ace2}}\label{app:ace2}
\begin{proof}
Given $\bp,\,\bq\, \in \mathcal{P}_n$. Let $\pi_{I}(\bp,\bq)$ and $\pi_{c}(\bp,\bq)$ denote the independent coupling and the comonotone coupling, respectively. Further, we use $\pi^{\downarrow}_{I}(\bp,\bq)$ and $\pi^{\downarrow}_{c}(\bp,\bq)$ to denote the independent and the comonotone couplings as vectors (PMFs) arranged in the non-increasing order.

Fix a $k \in [n^2]$. Consider the top-$k$ largest masses of the independent coupling $\pi^{}_{I}(\bp,\bq)$. Since $\bp$ and $\bq$ are sorted in the non-increasing order, for $i' \leq i, \, j' \leq j$ we have that $p_{i'}q_{j'}\geq p_{i}q_j$. Thus, the top-$k$ masses are concentrated towards the north-west corner of the matrix $\pi_{I}(\bp,\bq)$, consequently forming a Young's diagram. A Young's diagram is a left-justified arrangement of points in rows whose lengths form a non-increasing sequence. 

Let $\mathcal{D}$ denote the optimal set of top-$k$ masses i.e., $|\mathcal{D}|=k$, which is an Young's diagram. Now, let's say $\mathcal{D}$ uses `a' rows and `b' columns of the independent coupling matrix, then it completely lies inside the rectangular block $[a] \times [b]$. Thus,
\begin{align}
    \sum_{i=1}^{k}\pi^{\downarrow}_I(\bp,\bq)(i) &= \sum_{(i,j) \in \mathcal{D}}p_iq_j\\
    & \leq \sum_{(i,j) \in [a] \times [b]}p_iq_j\\
    &=\left(\sum_{i \leq a}p_i \right)\left(\sum_{j \leq b}q_j\right)\\
&\leq \min\left\{\sum_{i \leq a}p_i,\sum_{j \leq b}q_j \right\}    \label{eq:mass}
\end{align}
Now, we consider the same rectangular block $[a] \times [b]$ within the comonotone coupling matrix and show that it contains the exactly the mass in the RHS of~\eqref{eq:mass}.

Recall from the proof of Lemma~\ref{lem:ace1} that the comonotone coupling is the coupling between the random variables $(X,Y)$ where $X \sim \bp$ and $Y \sim \bq$ , generated in the following manner
\begin{align*}
    U &\triangleq \mathrm{Unif}([0,1])\\
    X &\triangleq \min\left\{i:\sum_{j=1}^{i}p_j \geq U\right\}  \\   Y &\triangleq \min\left\{i:\sum_{j=1}^{i}q_j \geq U\right\}
\end{align*}
Therefore, we have that
\begin{align}
    \pi_{c}(\bp,\bq)([a] \times [b]) &= \mathbb{P}(X \leq a, Y \leq b)\\
    &=\mathbb{P}\left( U \leq \sum_{j=1}^{a}p_j, U \leq \sum_{j=1}^{b}q_j \right)\\
    &=\min\left\{\sum_{i \leq a}p_i,\sum_{j \leq b}q_j \right\}\label{eq:mass2}
\end{align}
To complete our proof, we will now show that the rectangular block $[a] \times [b]$ in the comonotone coupling has at most $k$ non-zero elements. 

First, note that the comonotone coupling is supported on a monotone path through the matrix i.e., any north-east element or the south-west element around a non-zero element is zero. Thus, it starts from the north-west corner (top-left) of the matrix as the first nonzero cell, the next nonzero cell is obtained by moving $(i,j) \rightarrow (i+1,j)$ or $(i,j) \rightarrow (i,j+1)$. But if a row and column are exhausted simultaneously, the next nonzero can be $(i,j) \rightarrow (i+1,j+1)$. Consequently, the rectangular block $[a] \times [b]$ aligned in the north-west corner has at most $a+b-1$ non-zero elements. 

Recall the Young's diagram formed by the top-$k$ elements of the independent coupling matrix, which is tightly enclosed by the rectangular block $[a] \times [b]$. Therefore, $\mathcal{D}$ entirely contains the first row of length $b$ and the first column of length $a$, which implies that $k = |\mathcal{D}|\geq a+b-1$. As a result, using~\eqref{eq:mass} and~\eqref{eq:mass2} we have
\begin{align}
      \sum_{i=1}^{k}\pi^{\downarrow}_I(\bp,\bq)(i) &\leq \min\left\{\sum_{i \leq a}p_i,\sum_{j \leq b}q_j \right\}\\
      &\leq \sum_{i=1}^{a+b-1}\pi^{\downarrow}_c(\bp,\bq)(i)\\
      &\leq \sum_{i=1}^{k}\pi^{\downarrow}_c(\bp,\bq)(i)
\end{align}
Therefore, the partial sum of top-$k$ masses of the comonotone coupling is larger than the top-$k$ masses of the independent coupling. Since, the result holds for every $k \in [n^2]$, we conclude
\begin{align}
        \pi^{\downarrow}_I(\bp,\bq) \preceq \pi^{\downarrow}_c(\bp,\bq).
\end{align}
This completes our proof.
\end{proof}
\subsection{Proof of Theorem~\ref{thm:subadd_tsallis}}\label{app:subadd_tsallis}
\begin{proof}
From Lemma~\ref{lem:ace1}, we know that $  \pi_c(\bp,\bq) \sqsubseteq \bp \wedge \bq$. Since aggregation implies majorization, we have that
\begin{align}
\pi^{\downarrow}_c(\bp,\bq) \, \preceq \, \bp \wedge \bq
\end{align}    
Using Schur concavity of Tsallis entropy, we have that
\begin{align}
    T_{\alpha}(\bp \wedge\bq) \leq T_{\alpha}(\pi^{\downarrow}_c(\bp,\bq))\label{eq:llkd}
\end{align}
for every $\alpha \in [0,\infty)$.
\noindent Since, the Tsallis entropy only satisfies pseudoadditivity, we note that $T_{\alpha}(\pi^{\downarrow}_I(\bp,\bq)) \leq T_{\alpha}(\bp)+T_{\alpha}(\bq)$ holds only for $\alpha \geq 1$. Therefore, the Lemma~\ref{lem:ace2} can be leveraged to prove subadditivity only for $\alpha \geq 1$. However, we now use a different technique to prove our result bypassing the relation between the comonotone and the independent coupling, which works for every $\alpha \in [0,\infty)$.\\

\noindent Consider a PMF $\bp \in \mathcal{P}_n$. We have that 
\begin{align}
    T_{\alpha}(\bp)\,=\,\sum^{n}_{i=1}\dfrac{p_i-p_i^\alpha}{\alpha-1}\,=\,\sum^{n}_{i=1}h_{\alpha}(p_i)
\end{align}
where $h_{\alpha}(x):=\dfrac{x-x^{\alpha}}{\alpha-1}$.\\

Let $B_{\bp}$ and $B_\bq$ denote the sets of cumulative-sum breakpoints of the PMFs $\bp$ and $\bq$, respectively. Then, the common refinement of both the interval partitions is generated by the set $B_\bp \cup B_\bq$, and its interval lengths are the probability masses of the comonotone coupling $\pi_c(\bp,\bq)$.

Let $B$ be a finite set such that $B\subset(0,1)$, and
$0=b_0<b_1<\cdots<b_m<b_{m+1}=1$
be the ordered elements of $B\cup\{0,1\}$. Now, define
$g(B):=(b_1-b_0,b_2-b_1,\ldots,b_{m+1}-b_m)$.
Thus, $g(B)$ is a vector (more precisely, a PMF) of interval lengths induced by the breakpoints in $B$.
Now, let us define the function $
F_\alpha(B):=T_\alpha(g(B)).
$
Then, we note that
\begin{align}
F_\alpha(B_{\bp})=T_\alpha(\bp),
\qquad
F_\alpha(B_{\bq})=T_\alpha(\bq),
\end{align}
and from the definition of the comonotone coupling (cf. section~\ref{sec:preliminaries}), we have
\begin{align}
F_\alpha(B_{\bp}\cup B_{\bq}) =
T_\alpha(\pi_c^\downarrow(\bp,\bq)).
\end{align}
Therefore, we will now show that
\begin{align}
F_\alpha(B_{\bp}\cup B_{\bq})
\leq
F_\alpha(B_{\bp})+F_\alpha(B_{\bq}).
\end{align}

Assume first that $\alpha \in (0,\infty)$ and $\alpha \neq 1$. If an interval of length
$a+b$ is split into two intervals of lengths $a$ and $b$, then the increase in
$F_\alpha(\cdot)$ is
\begin{align}
\Delta_\alpha(a,b)
&=
h_\alpha(a)+h_\alpha(b)-h_\alpha(a+b)\\
&=
\frac{(a+b)^\alpha-a^\alpha-b^\alpha}{\alpha-1}
\end{align}
For $a,b>0$, we have that
\begin{align}
\frac{\partial}{\partial a}\Delta_\alpha(a,b)
=
\frac{\alpha\big((a+b)^{\alpha-1}-a^{\alpha-1}\big)}{\alpha-1}>0,
\end{align}
and similarly $\frac{\partial}{\partial b}\Delta_\alpha(a,b)>0$. Consequently, $\Delta_\alpha(a,b)$ is increasing in both arguments for every$\alpha \in (0,\infty)$, and $\alpha \neq 1$.

Now, let the cumulative-sum breakpoint set of $\bq$ be $B_{\bq}=\{x_1,\ldots,x_m\}$. Further, let
\begin{align}
A_t:=\{x_1,\ldots,x_t\},
\qquad
C_t:=B_{\bp}\cup A_t.
\end{align}
such that $A_m = B_{\bq}$. Furthermore, $C_m = B_{\bp} \cup B_{\bq}$ and $g(C_m)$ is the comonotone coupling vector, up to permutations.
Now, when the breakpoint $x_t$ is added to $A_{t-1}$, it splits some existing interval into
two intervals of lengths $a_t$ and $b_t$. When the same breakpoint $x_t$ is
added to $C_{t-1}$, we note that the interval being split is contained inside the former
interval. Thus, it is split into two intervals of lengths, say $a_t'$ and $b_t'$, where
\begin{align}
0\leq a_t'\leq a_t,
\qquad
0\leq b_t'\leq b_t.
\end{align}
Since $\Delta_\alpha$ is increasing in both arguments, we get
\begin{align}
F_\alpha(C_t)-F_\alpha(C_{t-1})
\leq
F_\alpha(A_t)-F_\alpha(A_{t-1}).
\end{align}
Summing over all breakpoints of $\bq$ i.e., $t=1,\ldots,m$, we obtain
\begin{align}
F_\alpha(B_{\bp}\cup B_{\bq})-F_\alpha(B_{\bp})
\leq
F_\alpha(B_{\bq})-F_\alpha(\varnothing).
\end{align}
But $F_\alpha(\varnothing)=T_\alpha((1))=0$. Therefore,
\begin{align}
F_\alpha(B_{\bp}\cup B_{\bq})
\leq
F_\alpha(B_{\bp})+F_\alpha(B_{\bq}).
\end{align}
Equivalently, we have that
\begin{align}
T_\alpha(\pi_c^\downarrow(\bp,\bq))
\leq
T_\alpha(\bp)+T_\alpha(\bq).
\end{align}
Combining this with~\eqref{eq:llkd}, we obtain
\begin{align}
T_\alpha(\bp\wedge\bq)
\leq
T_\alpha(\pi_c^\downarrow(\bp,\bq))
\leq
T_\alpha(\bp)+T_\alpha(\bq).
\end{align}
This proves subadditivity for $\alpha \in (0,\infty)$, $\alpha\neq1$.
The case $\alpha=1$ (Shannon entropy) follows by taking the limit $\alpha\to1$, which yields $h_1(x)=-x\log x$. Consequently, $\Delta_1(a,b)$ is increasing in both the arguments. A similar analysis as above recovers the already-known result for the Shannon entropy.\\

\noindent \textbf{(Proof for $\alpha =0$). } For $\alpha=0$, the Tsallis entropy reduces to,
\begin{align}
T_0(\mathbf r)=|\mathrm{supp}(\mathbf r)|-1.
\end{align}
If $a=|\mathrm{supp}(\bp)|$ and
$b=|\mathrm{supp}(\bq)|$,
then we have that
\begin{align}
|\mathrm{supp}(\bp\wedge\bq)|\le \max(a,b).
\end{align}
Consequently,
\begin{align}
T_0(\bp\wedge\bq)
\,\leq\,
\max(a,b)-1
\,\leq\,
a+b-2
\,=\,
T_0(\bp)+T_0(\bq),
\end{align}
as $a,\,b \geq 1$. Thus, the subadditivity property also holds for Tsallis entropy of order zero.\\

\noindent \textbf{(Conditions for equality). }
If either $\bp$ or  $\bq$ is deterministic, say
$\bp=(1,0,\ldots,0)$, then we have that
$ \bp \wedge \bq=\bq$.
Consequently, $T_\alpha(\bp)=0$ and the equality holds for all $\alpha \in [0,\infty)$.

Conversely, suppose neither $\bp$ nor $\bq$ is deterministic.
Then, both $B_{\bp}$ and $B_{\bq}$ are non-empty. Therefore, for
$\alpha \in (0,\infty)$, the first breakpoint of $B_{\bq}$, when added after
the breakpoints of $B_{\bp}$, gives a strictly smaller entropy gain than
when it is added to the unpartitioned interval $[0,1]$, i.e.,
\begin{align}
F_\alpha(B_{\bp}\cup B_{\bq})
<
F_\alpha(B_{\bp})+F_\alpha(B_{\bq}),
\end{align}
Consequently, we have
\begin{align}
T_\alpha(\bp\wedge\bq) < T_\alpha(\bp)+T_\alpha(\bq).
\end{align}
For $\alpha=0$, if neither distribution is deterministic, then $a,\,b\geq 2$, and
\begin{align}
T_0(\bp\wedge\bq)
\le
\max(a,b)-1
<
a+b-2
=
T_0(\bp)+T_0(\bq).
\end{align}
Therefore, the equality holds if and only if at least one of $\bp$ and
$\bq$ is a deterministic PMF. This completes our proof.
\end{proof}

\subsection{Proof of Theorem~\ref{thm:countersupermod}}\label{app:countersupermod}
\begin{proof}
We prove the result by illustrating two different pairs of PMFs in $\mathcal{P}_n$. The first pair violates
submodularity while the second pair violates supermodularity, for different values of $\alpha \in (0,1)$.

Given $\bp, \bq$ $\in \mathcal{P}_n$, let us define,
\begin{align*}
   \Delta_\alpha(\bp,\bq)
\triangleq
H_\alpha(\bp\wedge \bq)+H_\alpha(\bp\vee \bq)-H_\alpha(\bp)-H_\alpha(\bq). 
\end{align*}
Now, given an order $\alpha$, we note that supermodularity is equivalent to $\Delta_\alpha(\bp,\bq)\geq 0$ for all $\bp, \bq\in\mathcal P_n$, while submodularity is equivalent to $\Delta_\alpha(\bp,\bq)\le 0$ for all $\bp,\bq\in\mathcal P_n$.\\
\begin{figure}[H]
  \centering
  \input{example1}
  \caption{(Example~$1$) : Supermodular behavior of the Rényi entropy for order $\alpha \in (0,1)$.}\label{fig:example1}
  \vspace*{-3mm}
\end{figure}
\begin{figure}[H]
  \centering
  \input{example2}
  \caption{(Example~$2$) : Submodular behavior of the Rényi entropy for order $\alpha \in (0,1)$.}\label{fig:example2}
  \vspace*{-3mm}
\end{figure}
\noindent \textbf{Example $1$} - \emph{Pair with $\Delta_\alpha(\bp,\bq)>0$:} Let $\bp$, $\bq$ $\in$ $\mathcal{P}_3$ as,
\begin{align*}
    \bp &= (0.6,0.2,0.2)\\
    \bq &= (0.45,0.4,0.15)    
\end{align*}
Consequently, the glb and lub are as follows,
\begin{align*}
    \bp \wedge \bq &= (0.45,0.35,0.2)\\
    \bp \vee \bq &= (0.6,0.25,0.15)
\end{align*}
Therefore, we have 
\begin{align}
\begin{array}{c|c}
\alpha & \Delta_\alpha(\bp,\bq) \\
\hline
0 & 0 \\
0.2 & 0.00580417 \\
0.5 & 0.01387056 \\
0.7 & 0.01882329 \\
0.9 & 0.02343356 \\
1 & 0.02560746 \\
2 & 0.04204643 \\
\infty & 0
\end{array}
\end{align}
Therefore, for the above example we have that the Rényi entropy is not submodular for $\alpha \in (0,1)$.\\

\noindent \textbf{Example $2$} - \emph{Pair with $\Delta_\alpha(\bp,\bq)<0$:} Let $\bp$, $\bq$ $\in$ $\mathcal{P}_4$ as,
\begin{align*}
    \bp &= 0.398886918,0.370328848,0.228811150,0.001973084)\\
    \bq &= 0.539996140,0.229554617,0.116684354,0.113764889)   
\end{align*}
Consequently, the glb and lub are as follows,
\begin{align*}
    \bp &\wedge \bq \\&= (0.398886918,0.370328848,0.117019345,0.113764889)\\
    \bp &\vee \bq \\&= (0.539996140,0.229554617,0.228476159,0.001973084)
\end{align*}
Therefore, we have 
\begin{align}
\begin{array}{c|c}
\alpha & \Delta_\alpha(\bp,\bq) \\
\hline
0 & 0 \\
0.2 & -0.00090983 \\
0.5 & -0.00234206 \\
0.7 & -0.00198286 \\
0.9 & -0.00067566 \\
1 & 0.00022487 \\
2 & 0.00977669 \\
\infty & 0
\end{array}
\end{align}
Therefore, for the above example we have that the Rényi entropy is not supermodular for $\alpha \in (0,1)$.

Consequently, the first example gives $\Delta_\alpha(\bp,\bq)>0$ for $\alpha\in(0,1)$, and hence violates submodularity. The second example gives $\Delta_\alpha(\bp,\bq)<0$ for $\alpha\in(0,1)$, and hence violates supermodularity. Therefore, for every $\alpha\in(0,1)$, $H_\alpha$ is neither supermodular nor submodular on the majorization lattice $(\mathcal P_n,\preceq,\wedge,\vee)$. This completes the proof.
\end{proof}

\subsection{Proof of Theorem~\ref{thm:supermod_tsallis}}\label{app:supermod_tsallis}
\begin{proof}
Our proof method is based on studying the Lorenz curves of $\bp$, $\bq$, $\bp \wedge \bq$, and $\bp \vee \bq$. We handle the cases of $\alpha<1$ and $\alpha \geq 1$ separately below.\\

\noindent \textbf{(Proof for $\alpha<1$). }Let $S_{\alpha}(\bp):=\sum_{i=1}^{n} p^{\alpha}_i$, as defined in the proof of Theorem~\ref{thm:supermod}. Further, let $L_\bp$, $L_\bq$, $L_\bw$, $L_\bv$, $m$, and $M$ be also as defined in the proof of Theorem~\ref{thm:supermod}. Them, we have that, 
\begin{align}\label{eq:remt}
       \int^{n}_{0}(m'(t))^{\alpha}dt\;+ \int^{n}_{0}(M'(t))^{\alpha}dt\; = S_{\alpha}(\bp) +S_{\alpha}(\bq).
\end{align}
Now, fix an interval $I_k \triangleq (k-1,k)$. Then, $m$ is a piece-wise linear concave function within $I_k$ while $L_{\bw}$ is a straight line joining $m(k-1)$ to $m(k)$. Thus, 
\begin{align*}
    L'_{\bw}(t)&=\frac{m(k)-m(k-1)}{1}\\&=\int^{k}_{k-1}m'(t)dt
\end{align*}
Since $u \rightarrow u^{\alpha}$ is a concave map for $\alpha < 1$, using Jensen's inequality we obtain
\begin{align}\label{eq:main1t}
     (L'_{\bw}(t))^{\alpha} \geq \int^{k}_{k-1}m'(t)^{\alpha}dt
\end{align}
Now, we consider the function $M$ in the same interval $I_k$. Unlike $m$, $M$ is a piece-wise linear (not necessarily concave) function while $F_{\bv}$ is a straight line joining $M(k-1)$ to $M(k)$. If $I_k$ has no non-integer crossings, then $F'_{\bv}(t)=M'(t)$. 

However, if there are crossings then $M$ is convex and we need to concavify it to get $F_{\bv}$. Thus,
\begin{align*}
        F'_{\bv}(t)&=\frac{M(k)-M(k-1)}{1}\\&=\int^{k}_{k-1}M'(t)dt
\end{align*}
Since $u \rightarrow u^{\alpha}$ is a concave map for $\alpha<1$, using Jensen's inequality we have
\begin{align}\label{eq:main2t}
     (F'_{\bv}(t))^{\alpha} \geq \int^{k}_{k-1}(M'(t))^{\alpha}dt
\end{align}
Now, if $F_{\bv}$ is a piecewise-linear concave function, then it is the lub $L_{\bv}$. Otherwise, it needs to be concavified and consequently the the lub $L_{\bv}$ is the least concave majorant (LCM) of $F_{\bv}$.

Therefore applying Jensen again over the pooled intervals used by the LCM operation, we get
\begin{align}\label{eq:main3t}
S_\alpha(\bv)
    &=
    \int_0^n (L_{\bv}'(t))^\alpha\,dt\\
    &\geq
    \int_0^n (F_{\mathbf v}'(t))^\alpha\,dt
    \geq
    \int_0^n (M'(t))^\alpha\,dt .
\end{align}
On summing~\eqref{eq:main1t} over all intervals $I_k$ for $k \in [n]$ and using~\eqref{eq:main3t}, we obtain
\begin{align}
    S_{\alpha}(\bw)+S_{\alpha}(\bv) \geq \int^{n}_{0}(m'(t))^{\alpha}dt +\int^{n}_{0}(M'(t))^{\alpha}dt 
\end{align}
Therefore, from~\eqref{eq:remt} we have
\begin{align}\label{eq:rem2t}
     S_{\alpha}(\bw)+S_{\alpha}(\bv) \geq S_{\alpha}(\bp)+S_{\alpha}(\bq).
\end{align}
Recall that $T_{\alpha}(\bx):=\frac{1-S_{\alpha}(\bx)}{\alpha-1}$, therefore we have that
\begin{align}
    1-S_{\alpha}(\bw)+1-S_{\alpha}(\bv) \leq 1-S_{\alpha}(\bp)+1-S_{\alpha}(\bq)
\end{align}
Since $\alpha -1 < 0$, on dividing both sides by $\alpha-1$ we obtain
\begin{align}
    \frac{1-S_{\alpha}(\bw)}{\alpha-1}+\frac{1-S_{\alpha}(\bv)}{\alpha-1} \geq \frac{1-S_{\alpha}(\bp)}{\alpha-1}+\frac{1-S_{\alpha}(\bq)}{\alpha-1}    
\end{align}
Therefore, 
\begin{align}
    T_{\alpha}(\bp)+T_{\alpha}(\bq) \leq T_{\alpha}(\bw)+T_{\alpha}(\bv).
\end{align}
This completes our proof and shows that Tsallis entropy is supermodular for the case of $\alpha < 1$.\\

\noindent \textbf{(Proof for $\alpha \geq 1$). } The proof for the case of $\alpha>1$ directly follows from the proof of Theorem~\ref{thm:supermod}. Let $S_{\alpha}(\bp):=\sum_{i=1}^{n} p^{\alpha}_i$, as defined in the proof of Theorem~\ref{thm:supermod}. Thus, from~\eqref{eq:rem2} we know that for $\alpha \geq 1$,
\begin{align}
         S_{\alpha}(\bw)+S_{\alpha}(\bv) \leq S_{\alpha}(\bp)+S_{\alpha}(\bq),
\end{align}
where $\bw:=\bp \wedge \bq$ and $\bv:=\bp \vee \bq$. Recall that $T_{\alpha}(\bx):=\frac{1-S_{\alpha}(\bx)}{\alpha-1}$, therefore we have that
\begin{align}
    1-S_{\alpha}(\bw)+1-S_{\alpha}(\bv) \geq 1-S_{\alpha}(\bp)+1-S_{\alpha}(\bq)
\end{align}
Since $\alpha -1 > 0$, on dividing both sides by $\alpha-1$ we obtain
\begin{align}
    \frac{1-S_{\alpha}(\bw)}{\alpha-1}+\frac{1-S_{\alpha}(\bv)}{\alpha-1} \geq \frac{1-S_{\alpha}(\bp)}{\alpha-1}+\frac{1-S_{\alpha}(\bq)}{\alpha-1}    
\end{align}
Therefore, 
\begin{align}
    T_{\alpha}(\bp)+T_{\alpha}(\bq) \leq T_{\alpha}(\bw)+T_{\alpha}(\bv).
\end{align}
This completes our proof and shows that Tsallis entropy is supermodular for the case of $\alpha > 1$. The $\alpha \to 1$ case follows due to the continuity of Tsallis entropy in order $\alpha$.\\

\noindent \textbf{(Modularity for $\alpha = 0$). }Let
$ a := |\mathrm{supp}(\bp)|$ and $b := |\mathrm{supp}(\bq)|$,
where $a,b\, \leq \, n$. For $\alpha\rightarrow0$, we have
\begin{align*}
    |\mathrm{supp}(\bp\wedge\bq)|=\max(a,b),
    \quad
    |\mathrm{supp}(\bp\vee\bq)|=\min(a,b).
\end{align*}
Therefore,
\begin{align}
    T_0(\bp\wedge\bq)
    +
    T_0(\bp\vee\bq)
    &=
    \max(a,b)
    +
    \min(a,b) -2\\
    &=
    a + b-2\\
    &=
    T_0(\bp)+T_0(\bq).
\end{align}
This completes our proof.
\end{proof}
\end{document}